\documentclass[journal,10pt,a4paper,twoside]{IEEEtran}

\usepackage{amsthm}
\usepackage{amsmath,amssymb,amsfonts,mathtools}
\usepackage{bm}
\usepackage{dsfont}
\usepackage{physics}

\usepackage{graphicx}
\usepackage{booktabs}
\usepackage{subcaption}
\usepackage{acronym}
\usepackage{textcomp}
\usepackage{xcolor}
\usepackage{cite}
\usepackage{newpxtext,newpxmath}
\usepackage[hidelinks]{hyperref}

\definecolor{ocre}{HTML}{800000}
\definecolor{green}{RGB}{0,128,0}
\definecolor{sky}{HTML}{C6D9F1}
\definecolor{skybox}{HTML}{5F86B3}

\newtheorem{proposition}{Proposition}

\title{Quantum MIMO Channel Modeling in Turbulent Free-Space Optical Links}
\author{Heyang Peng, Seid Koudia, Semih Oktay, Mert Bayraktar, and Symeon Chatzinotas, Fellow, IEEE}

\begin{document}
\maketitle
\begin{abstract}
Free-space optical (FSO) links supporting spatial multiplexing provide a natural physical realization of Quantum MIMO channels. We develop a first-principles model for Quantum MIMO channels derived directly from wave-optical propagation through three-dimensional atmospheric turbulence. The framework explicitly accounts for intermodal crosstalk, finite detection apertures, and the system–bath separation induced by spatial-mode projection.
We distinguish between distinguishable and indistinguishable photon regimes, showing that indistinguishability leads to intrinsically many-body interference effects described by matrix permanents. To obtain a completely positive and trace-preserving logical description, we introduce an erasure-extended encoding in which turbulence-induced leakage and photon loss are mapped to flagged erasure states. The resulting Quantum MIMO channel naturally reduces to a correlated $n$-qubit erasure channel, with correlations arising from the shared turbulent medium. Limiting regimes in which correlated Pauli channels emerge as effective approximations are also identified.
\end{abstract}

%
\section{Introduction}

 Quantum communication aims to exploit the fundamental principles of quantum mechanics to enable secure information transfer \cite{koudia2019superposition}, distributed quantum computing, and quantum networking over long distances. Free-space optical (FSO) links are a promising platform for quantum communication due to their ability to support high-dimensional spatial multiplexing, long-distance transmission, and flexible network geometries \cite{11317988}. When multiple spatial modes are transmitted simultaneously through a turbulent atmosphere, the resulting system is naturally interpreted as a quantum multiple-input multiple-output (Quantum MIMO) channel \cite{junaid2025diversity,ur2025mimo,koudia2025crosstalk,koudia2025spatial}. Each spatial mode functions as a distinct transmission ``\textcolor{black}{source}'', while atmospheric turbulence and finite-aperture detection induce coupling, loss, and correlations among the transmitted degrees of freedom \cite{pirandola2021limits}.

In this work, we focus on a physically relevant operating regime in which each spatial mode carries exactly one photon whose polarization encodes a qubit \cite{fabre2020modes, peng2025performance,pengqce}. Logical information is therefore stored in an $n$-qubit polarization register, while the spatial degrees of freedom play the role of a structured environment. Turbulence-induced intermodal coupling \textcolor{black}{redistributes optical power among spatial modes during propagation, with a fraction of the field scattered into spatial modes outside the finite set collected at the receiver; projection onto the retained modes therefore induces loss and renders the reduced polarization dynamics intrinsically non-unitary.
}\footnote{\textcolor{black}{In this work, “spatial modes” refer to orthogonal transverse field modes of the optical field (e.g., Laguerre–Gaussian modes), rather than to physical emitting elements or apertures. A single transmitter aperture may excite multiple spatial modes simultaneously (for example using a spatial light modulator or mode multiplexer), and a single receiver aperture may project onto multiple modes using mode sorting or matched filtering. Conversely, multiple elements that excite the same spatial mode do not constitute independent channels in the modal sense, as they couple to a single orthogonal degree of freedom. Independent channels can instead arise either from orthogonal spatial modes within a single aperture, or from physically separated transmit apertures. The latter case corresponds to spatial multiplexing in the classical sense and may lead to partially independent channels depending on the aperture separation relative to the turbulence coherence length.}}.

Unlike phenomenological noise models commonly employed in quantum information theory, the effective channel acting on polarization qubits in an FSO link is not assumed \emph{a priori}. Instead, it emerges from the underlying wave-optical propagation \cite{Ishimaru1999,Goldsmith2005}, the statistics of the turbulent medium, and the operational choice of transmitter and receiver mode bases. As a result, the induced noise can exhibit nontrivial features, including correlations across qubits, effective non-Markovian behavior along the propagation direction, and qualitative dependence on photon indistinguishability through multi-photon interference effects.

The goal of this paper is to construct a Quantum MIMO channel model directly from first principles. Starting from a wave-optical description of multi-mode propagation through atmospheric turbulence, we derive the effective quantum channel acting on the polarization degrees of freedom after spatial-mode projection and loss. This framework makes explicit how intermodal crosstalk, partial mode collection, and bosonic interference jointly determine the structure of the resulting channel. In particular, we clarify the distinction between distinguishable and indistinguishable photon regimes, show how erasure must be incorporated to obtain a completely positive trace-preserving description, and connect the resulting physically derived Quantum MIMO channels to standard correlated $n$-qubit noise models used in quantum information theory \cite{Wilde2017}.

\section{Model: Quantum MIMO Free-Space Optical Channel}
\begin{figure*}[t]
\centering
\includegraphics[width = 0.9\textwidth]{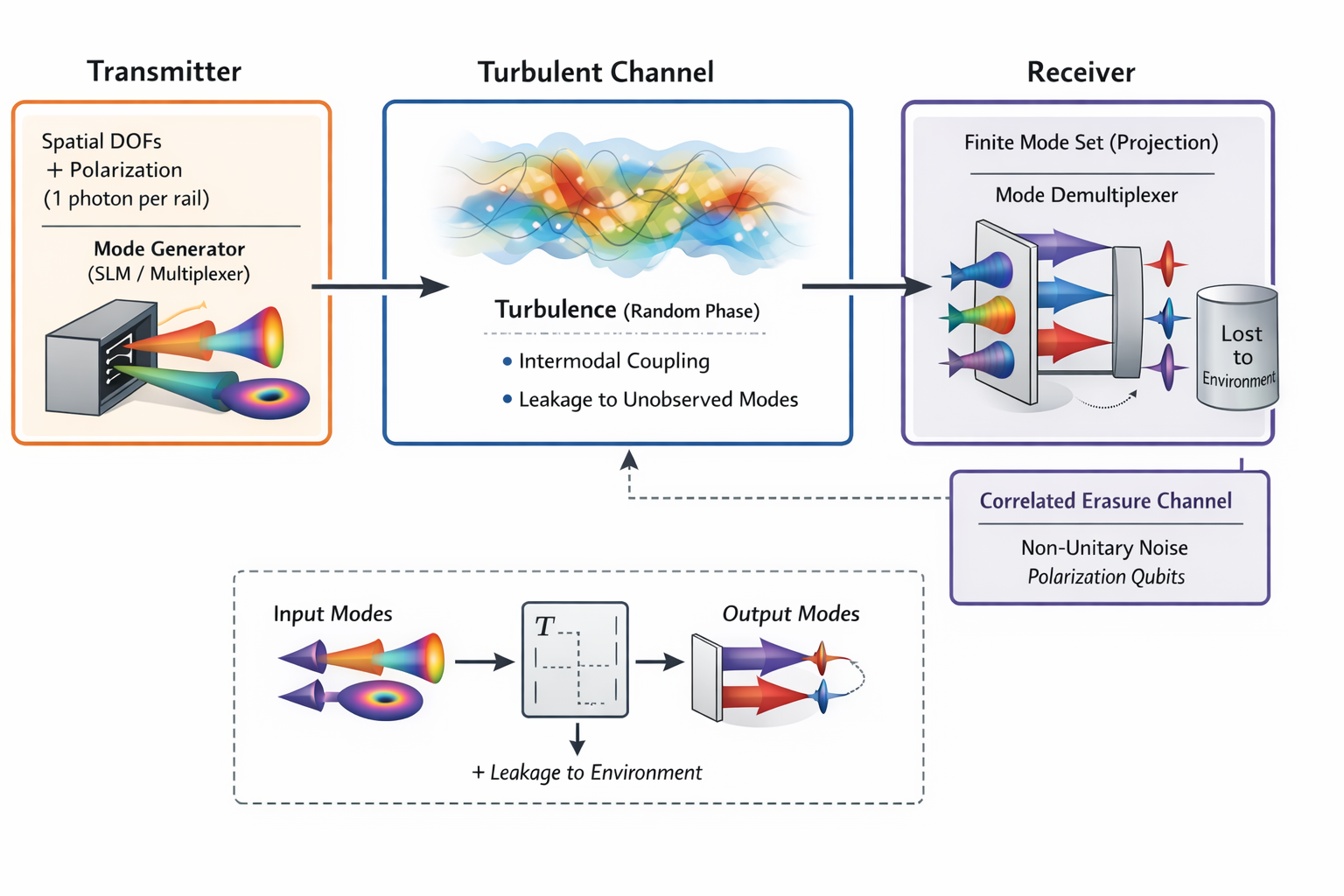}
\caption{\color{black} Schematic of a spatially multiplexed FSO quantum link. Parallel channels are defined by orthogonal spatial degrees of freedom carrying polarization-encoded qubits. Propagation through atmospheric turbulence leads to mixing between channels and partial loss outside the collected set, while the receiver performs a finite-mode projection, giving rise to effective noise in the logical system.}
	\label{fig:general_MIMO}
\end{figure*}
The overall system considered in this work is illustrated in Fig. \ref{fig:general_MIMO}. Multiple orthogonal spatial degrees of freedom, which may be realized as spatial modes or partially separated beams, are used to transmit polarization-encoded qubits in parallel. During propagation, atmospheric turbulence induces coupling between these channels, which can be described by a random mode-mixing operator 
T
T, and scatters optical power into unobserved spatial degrees of freedom. At the receiver, projection onto a finite set of modes defines the effective system subspace, while the remaining degrees of freedom are treated as an environment. This system–bath separation leads naturally to a reduced, generally non-unitary quantum channel acting on the logical polarization degrees of freedom.
\subsection{Quantum MIMO Encoding and Degrees of Freedom}

We consider a single-shot free-space optical transmission in which $n$ orthogonal spatial modes are launched simultaneously from a transmitter plane at $z=0$ \cite{Allen1992,Willner2015}. Each spatial mode carries exactly one photon, and the logical information is encoded in the polarization degree of freedom. This realizes a Quantum MIMO channel, where spatial modes play the role of parallel physical channels.

\color{black}
We let the transmitted spatial modes to be Laguerre--Gaussian (LG) modes, which form a complete orthonormal set of solutions to the paraxial wave equation in cylindrical coordinates $(r,\varphi,z)$. LG modes are labeled by a radial index $p \in \mathbb{N}_0$ and an azimuthal index $\ell \in \mathbb{Z}$, where $\ell$ determines the orbital angular momentum carried by the mode. In this work, we restrict attention to the lowest-radial-order family $p=0$, which is commonly employed in free-space optical systems due to its high mode purity and experimental robustness. At the transmitter plane $z=0$, the normalized transverse field profile of an LG mode with azimuthal index $\ell$ is given by
\color{black}
\[
u_{\ell}(r,\varphi)=
L_0^{(|\ell|)}\!\left(\frac{2r^2}{w_0^2}\right)
\exp\!\left(-\frac{r^2}{w_0^2}\right)
\exp(i\ell\varphi)
\]
with
\color{black}
\[
L_0^{(|\ell|)}
=
\sqrt{\frac{2}{\pi |\ell|!}}\,
\frac{1}{w_0}
\left(\frac{\sqrt{2}\, r}{w_0}\right)^{|\ell|}
\]
where $w_0$ is the beam waist at the transmitter plane. These modes satisfy the orthonormality relation
$\int d^2 r\, u_\ell^*(r,\varphi)\,u_{\ell'}(r,\varphi) = \delta_{\ell\ell'}$
and each photon in mode $\ell$ carries orbital angular momentum $\ell\hbar$.

\color{black}
Each spatial mode supports two orthogonal polarizations, labeled $H$ and $V$, which encode a qubit:
\begin{equation}
\ket{0}_m \equiv \ket{H}_m, \qquad \ket{1}_m \equiv \ket{V}_m.
\end{equation}
The logical polarization Hilbert space is therefore
\begin{equation}
\mathcal{H}_{\mathrm{pol}} = (\mathbb{C}^2)^{\otimes n}.
\end{equation}

The full optical field occupies a much larger Hilbert space, including all transverse spatial modes outside the chosen set and both polarizations. These additional degrees of freedom constitute an effective environment (bath) that is traced out after propagation and spatial-mode projection.

\subsection{Atmospheric Turbulence Model}

Atmospheric turbulence is modeled as a three-dimensional, statistically homogeneous refractive-index fluctuation field $n_1(\mathbf{r},z)$ with zero mean and \textcolor{black}{modified} von K\'arm\'an power spectrum \cite{AndrewsPhillips200}: 
\begin{equation}
\Phi_n(\boldsymbol{\kappa}) =
0.033\, C_n^2
\left(\kappa^2 + \kappa_0^2\right)^{-11/6}
\exp\!\left(-\frac{\kappa^2}{\kappa_m^2}\right),
\end{equation}
where $\boldsymbol{\kappa}=(\kappa_x,\kappa_y,\kappa_z)$, $\kappa_0=2\pi/L_0$ is set by the outer scale $L_0$, and $\kappa_m=5.92/l_0$ is set by the inner scale $l_0$.

For \textcolor{black}{refractive index structure constant} $C_n^2$ along a propagation distance $L$, the turbulence strength is equivalently characterized by the Fried parameter and the Rytov variance,
\begin{equation}
r_0(L) = \bigl(0.423\,k_0^2 C_n^2 L\bigr)^{-3/5},
\qquad
\sigma_R^2(L) = 1.23\,C_n^2 k_0^{7/6} L^{11/6}.
\end{equation}

The three-dimensional nature of $n_1(\mathbf{r},z)$ induces longitudinal correlations along the propagation direction, which play a crucial role in generating correlated noise in the Quantum MIMO channel.

\subsection{Split-Step Paraxial Propagation and Phase Screens}

We discretize the propagation distance into $K$ slabs of thickness $\Delta z = L/K$, with planes at $z_k = k\Delta z$. The slab-integrated phase screen is
\begin{equation}
\phi_k(\mathbf{r}) =
k_0 \int_{z_k}^{z_{k+1}} n_1(\mathbf{r},z)\,dz
\approx k_0 \Delta z\, n_1(\mathbf{r},z_k).
\end{equation}

Propagation through each slab is modeled by the split-step operator
\begin{equation}
\mathcal{U}_{k+1:k} = \mathcal{P}_{\Delta z} \circ \mathcal{S}_{\phi_k}
\label{eq:propagator}
\end{equation}
where
\begin{equation}
(\mathcal{S}_{\phi_k} f)(\mathbf{r}) = e^{i\phi_k(\mathbf{r})} f(\mathbf{r})
\end{equation}
and $\mathcal{P}_{\Delta z}$ is the Fresnel propagator,
\begin{equation}
(\mathcal{P}_{\Delta z}f)(\mathbf{r}) =
\frac{e^{ik_0\Delta z}}{i\lambda \Delta z}
\int d^2\mathbf{r}'\,
\exp\!\left(\frac{ik_0}{2\Delta z}
\|\mathbf{r}-\mathbf{r}'\|^2\right) f(\mathbf{r}').
\end{equation}

Atmospheric turbulence is assumed polarization-independent in the bulk, so both $H$ and $V$ components experience identical spatial evolution.
\color{black}
Accordingly, benefiting from Fourier transform property of the indefinite integrals Huygens-Fresnel integral turns into \cite{GoodmanFourierOptics}
\color{black}: 
\begin{equation}
\begin{aligned}
u_r(r_x,r_y,L)
&= \exp(jkL)\,
\mathcal{F}^{-1}\!\Big[
\mathcal{F}(s-f)\,
\mathcal{F}\{u_{\ell}(r,\varphi)\}\,
\mathcal{F}(s-f) \\
&\qquad \times
\exp\!\left(\frac{jk}{2L}(s_x^2+s_y^2)\right)
\Big] \\
&= \exp(jkL)\,
\mathcal{F}^{-1}\!\left\{
u_{\ell}(f)\,
\exp\!\left(-\frac{j2\pi L}{k}f_r^2\right)
\right\} \\
&= \exp(jkL)\,
\mathcal{F}^{-1}\!\left\{
u_{\ell}(f_x,f_y)\,
\exp\!\left(-\frac{j2\pi L}{k}
\left(f_x^2+f_y^2\right)\right)
\right\}.
\end{aligned}
\end{equation}

The operators $\mathcal{F}$ and $\mathcal{F}^{-1}$
denote the two-dimensional Fourier transform and its inverse, respectively.
In the first line of Eq.~(9), the coordinate transformations are explicitly
indicated adjacent to the Fourier operators $\mathcal{F}$ and
$\mathcal{F}^{-1}$ for clarity.
The function $U_s(f)=U_s(f_x,f_y)$ corresponds
to the spatial-frequency representation of the source field in Cartesian coordinates, where
$(f_x,f_y)$ are the transverse spatial frequencies associated with the
wavelength $\lambda$. 
The field $u_r(r,L)=u_r(r_x,r_y,L)$ represents the complex optical field
evaluated at the receiver plane, which is located at a propagation distance
$L$ from the source plane. The transverse spatial coordinates at the receiver in Cartesian coordinates
are denoted by $r=(r_x,r_y)$.

\color{black}
\subsection{Inter-Mode Crosstalk as a Physical Mechanism}

The action of $\mathcal{U}_{k+1:k}$ on the spatial field causes power initially in one LG mode to scatter into other modes. This inter-mode coupling arises because the random phase $\phi_k(\mathbf{r})$ breaks the orthogonality of the LG basis. As a result, after propagation and projection onto a finite set of modes at the receiver, energy initially confined to a given spatial channel is redistributed among all channels.

This turbulence-induced mode mixing is the physical origin of spatial crosstalk in the Quantum MIMO channel and leads to correlated noise acting on the polarization qubits.

\subsection{Single-particle mode mixing and intermodal crosstalk}

Let $\{ w_m^{(k)}(\mathbf{r}) \}_{m=1}^n$ denote an orthonormal set of spatial modes defining the kept subspace at longitudinal plane $z_k$. The annihilation operator for spatial mode $m$ and polarization $p \in \{H,V\}$ at plane $z_k$ is
\begin{equation}
    \hat a_{m,p}^{(k)} =
    \int d^2\mathbf{r}\, w_m^{(k)*}(\mathbf{r}) \hat E_p(\mathbf{r},z_k),
\end{equation}
where $\hat E_p(\mathbf{r},z)$ is the positive-frequency field operator.

Because the propagation is linear and passive, the transformation from plane $z_k$ to $z_{k+1}$ induces a unitary mixing on the full set of spatial modes. Restricting attention to the kept subspace and its complement, the mode operators transform as
\begin{equation}
    \begin{pmatrix}
        \hat{\bm a}_S^{(k+1)} \\
        \hat{\bm a}_E^{(k+1)}
    \end{pmatrix}
    =
    \begin{pmatrix}
        T_k & R_k \\
        \tilde R_k & \tilde T_k
    \end{pmatrix}
    \begin{pmatrix}
        \hat{\bm a}_S^{(k)} \\
        \hat{\bm a}_E^{(k)}
    \end{pmatrix},
\end{equation}
where $\hat{\bm a}_S^{(k)}$ collects the $2n$ system-mode operators $(m,p)$ and $\hat{\bm a}_E^{(k)}$ collects the bath modes.

The $2n \times 2n$ matrix $T_k$ captures intermodal crosstalk within the kept subspace, including polarization mixing if present. The matrix $R_k$ captures leakage into bath modes. Unitarity of the full transformation implies
\begin{equation}
    T_k T_k^\dagger + R_k R_k^\dagger = \mathbb{I}_{2n}.
\end{equation}

At this stage, the description is purely single-particle and applies equally to distinguishable and indistinguishable photons. The distinction between these regimes arises when lifting this transformation to the many-body Hilbert space.

\section{Results: Explicit Derivation of the Quantum MIMO Channel}

\subsection{From Turbulence to the Single-Photon Propagation Operator}

We begin from the scalar paraxial wave equation governing the slowly varying envelope
$U(\mathbf{r},z)$ of a monochromatic optical field propagating through a weakly
fluctuating refractive index
$n(\mathbf{r},z)=n_0+n_1(\mathbf{r},z)$:
\begin{equation}
\left(
\frac{\partial}{\partial z}
- \frac{i}{2k_0}\nabla_\perp^2
\right)
U(\mathbf{r},z)
=
i k_0 n_1(\mathbf{r},z)\,U(\mathbf{r},z),
\label{eq:paraxial_eq}
\end{equation}
where $k_0=2\pi/\lambda$ and $\nabla_\perp^2$ is the transverse Laplacian.

We discretize the propagation interval $[0,L]$ into $K$ slabs of thickness
$\Delta z=L/K$, with planes at $z_k=k\Delta z$. Over a single slab, we neglect
longitudinal diffraction inside the slab and integrate
\eqref{eq:paraxial_eq} to first order in $\Delta z$, obtaining the split-step
approximation
\begin{equation}
U(\mathbf{r},z_{k+1})
=
\mathcal{P}_{\Delta z}
\!\left[
e^{i\phi_k(\mathbf{r})} U(\mathbf{r},z_k)
\right],
\label{eq:split_step_solution}
\end{equation}
where the random phase screen is
\begin{equation}
\phi_k(\mathbf{r})
=
k_0 \int_{z_k}^{z_{k+1}} n_1(\mathbf{r},z)\,dz
\approx k_0 \Delta z\, n_1(\mathbf{r},z_k),
\label{eq:phase_screen}
\end{equation}
and the Fresnel propagator $\mathcal{P}_{\Delta z}$ is
\begin{equation}
(\mathcal{P}_{\Delta z}f)(\mathbf{r})
=
\frac{e^{ik_0\Delta z}}{i\lambda\Delta z}
\int d^2\mathbf{r}'\,
\exp\!\left[
\frac{i k_0}{2\Delta z}
\|\mathbf{r}-\mathbf{r}'\|^2
\right] f(\mathbf{r}').
\label{eq:fresnel_kernel}
\end{equation}

Because $n_1(\mathbf{r},z)$ is drawn from a single three-dimensional von K\'arm\'an
random field, the random variables $\{\phi_k(\mathbf{r})\}$ are statistically
correlated along $k$. This longitudinal correlation is the microscopic origin of
correlated noise in the resulting Quantum MIMO channel.

Defining the single-photon transverse propagation operator by Eq. \ref{eq:propagator},
the full propagation from $z_0=0$ to $z_k$ is
\begin{equation}
\mathcal{U}_{k:0}
=
\mathcal{U}_{k:k-1}\circ\cdots\circ\mathcal{U}_{1:0}.
\end{equation}

Atmospheric turbulence is assumed polarization independent in the bulk, so
$\mathcal{U}_{k+1:k}$ acts identically on both polarization components. The
polarization degree of freedom therefore factors out at the single-photon level
and is only affected indirectly through spatial-mode mixing and projection, as
shown in the subsequent derivations.

\subsection{From Propagation to Mode Overlap Coefficients and Inter-Mode Crosstalk}

We now connect the split-step propagation operator $\mathcal{U}_{k+1:k}$ to the
inter-mode crosstalk coefficients that define the Quantum MIMO coupling matrix.

\subsubsection{Kept spatial mode bases at intermediate planes}

At each plane $z_k$, we define an orthonormal set of $n$ \emph{kept} spatial modes
$\{w_m^{(k)}(\mathbf{r})\}_{m=1}^n \subset L^2(\mathbb{R}^2)$ satisfying
\begin{equation}
\int d^2\mathbf{r}\, w_m^{(k)*}(\mathbf{r}) w_{m'}^{(k)}(\mathbf{r})
=
\delta_{mm'}.
\label{eq:kept_orthonormal}
\end{equation}
This set defines the operational ``system'' spatial subspace at plane $z_k$.
We complete it to a full orthonormal basis
$\{w_\mu^{(k)}(\mathbf{r})\}_{\mu\ge 1}$ of $L^2(\mathbb{R}^2)$ such that
$\mu\le n$ corresponds to kept modes and $\mu>n$ corresponds to bath modes.

\subsubsection{Single-photon overlap coefficients}

Consider a single transverse field mode $w_{\mu'}^{(k)}$ at plane $z_k$. After one
split-step slab, the transverse field at plane $z_{k+1}$ is
\begin{equation}
w_{\mu'}^{(k)} \ \mapsto\ \mathcal{U}_{k+1:k} w_{\mu'}^{(k)}.
\end{equation}
Expanding the propagated mode in the output basis $\{w_\mu^{(k+1)}\}$ gives
\begin{equation}
\mathcal{U}_{k+1:k} w_{\mu'}^{(k)}
=
\sum_{\mu\ge 1} t^{(k)}_{\mu\leftarrow \mu'}\, w_\mu^{(k+1)},
\label{eq:mode_expansion}
\end{equation}
where the expansion coefficients are the overlap integrals
\begin{equation}
t^{(k)}_{\mu\leftarrow \mu'}
=
\bigl\langle
w_\mu^{(k+1)},
\mathcal{U}_{k+1:k} w_{\mu'}^{(k)}
\bigr\rangle
=
\int d^2\mathbf{r}\,
w_\mu^{(k+1)*}(\mathbf{r})\,
(\mathcal{U}_{k+1:k} w_{\mu'}^{(k)})(\mathbf{r}).
\label{eq:t_coeff_def}
\end{equation}
Equation \eqref{eq:t_coeff_def} is the fundamental optical quantity that encodes
turbulence-induced spatial mixing.

\subsubsection{Explicit dependence on the phase screens and Fresnel kernel}

Substituting the split-step form
$\mathcal{U}_{k+1:k}=\mathcal{P}_{\Delta z}\circ\mathcal{S}_{\phi_k}$
into \eqref{eq:t_coeff_def} yields an explicit integral expression:
\begin{align}
t^{(k)}_{\mu\leftarrow \mu'}
&=
\int d^2\mathbf{r}\,
w_\mu^{(k+1)*}(\mathbf{r})
(\mathcal{P}_{\Delta z}[e^{i\phi_k} w_{\mu'}^{(k)}])(\mathbf{r})
\nonumber\\
&=
\int d^2\mathbf{r}\,
w_\mu^{(k+1)*}(\mathbf{r})\,
\frac{e^{ik_0\Delta z}}{i\lambda\Delta z}
\nonumber\\
&\quad \times \int d^2\mathbf{r}'\,
\exp\!\left[
\frac{i k_0}{2\Delta z}\|\mathbf{r}-\mathbf{r}'\|^2
\right]
e^{i\phi_k(\mathbf{r}')}
w_{\mu'}^{(k)}(\mathbf{r}').
\label{eq:t_coeff_explicit}
\end{align}
Thus $t^{(k)}_{\mu\leftarrow \mu'}$ depends on:
(i) the Fresnel kernel (hence $\lambda$ and $\Delta z$),
(ii) the random phase screen $\phi_k(\mathbf{r})$ (hence the turbulence field),
and (iii) the chosen kept/bath basis functions at planes $k$ and $k+1$.

Since $\phi_k$ is itself determined by the von K\'arm\'an model parameters
$(C_n^2,L_0,l_0)$ (or $(r_0,\sigma_R^2)$), Eq.~\eqref{eq:t_coeff_explicit} provides
a direct optical route from turbulence statistics to inter-mode crosstalk
coefficients.

\subsubsection{Inter-mode crosstalk in the kept subspace}

Define the \emph{kept-to-kept} spatial crosstalk matrix for the slab $k$ as
\begin{equation}
(T^{\perp}_k)_{m m'}
:=
t^{(k)}_{m\leftarrow m'}
=
\left\langle
w_m^{(k+1)},
\mathcal{U}_{k+1:k} w_{m'}^{(k)}
\right\rangle,
\qquad m,m'\in\{1,\dots,n\}.
\label{eq:Tperp_def}
\end{equation}
This $n\times n$ matrix describes how energy initially in kept mode $m'$ at plane
$z_k$ is redistributed among kept modes at plane $z_{k+1}$. The off-diagonal
entries $m\neq m'$ quantify inter-mode crosstalk.

Similarly, the \emph{kept-to-bath} leakage coefficients are
\begin{equation}
(R^{\perp}_k)_{\mu m'}
:=
t^{(k)}_{\mu\leftarrow m'},
\qquad \mu>n,\ m'\le n,
\end{equation}
which quantify scattering out of the kept subspace.

\subsubsection{A unitarity identity implying contraction on the kept subspace}

Because $\mathcal{U}_{k+1:k}$ is a unitary operator on $L^2(\mathbb{R}^2)$, the
full infinite-dimensional matrix $[t^{(k)}_{\mu\leftarrow\mu'}]_{\mu,\mu'}$ is
unitary. Restricting to the kept subspace implies the contraction identity
\begin{equation}
T_k^{\perp} (T_k^{\perp})^\dagger + R_k^{\perp} (R_k^{\perp})^\dagger = I_n,
\label{eq:contraction_spatial}
\end{equation}
where $R_k^{\perp}$ denotes the rectangular matrix with entries
$(R_k^{\perp})_{\mu m'}=t^{(k)}_{\mu\leftarrow m'}$ for $\mu>n$ and $m'\le n$.
Equation \eqref{eq:contraction_spatial} formalizes that loss/leakage to the bath
renders the effective kept-subspace evolution non-unitary.

The next step is to include polarization and second quantization, yielding the
$2n\times 2n$ system matrix $T_k$ acting on polarization-resolved mode operators
and setting the stage for an explicit Kraus construction on the logical
$n$-qubit space.

\subsection{Construction of the Polarization-Resolved System Matrix $T_k$}

We now derive the $2n\times 2n$ system matrix $T_k$ that governs the input--output
relation of polarization-resolved annihilation operators in the kept subspace.

\subsubsection{Polarization-resolved mode operators}

Let $\hat E_p(\mathbf{r},z)$ denote the positive-frequency field operator envelope
for polarization $p\in\{H,V\}$. For each plane $z_k$, define the annihilation
operator for spatial mode $\mu$ and polarization $p$ by
\begin{equation}
\hat a^{(k)}_{\mu,p}
=
\int d^2\mathbf{r}\,
w_\mu^{(k)*}(\mathbf{r})\,\hat E_p(\mathbf{r},z_k).
\label{eq:annihilator_def}
\end{equation}
We collect the $2n$ kept operators into a column vector
\begin{equation}
\hat{\bm a}^{(k)}_S
=
\bigl(
\hat a^{(k)}_{1,H},\hat a^{(k)}_{1,V},\dots,\hat a^{(k)}_{n,H},\hat a^{(k)}_{n,V}
\bigr)^{T},
\end{equation}
and similarly define $\hat{\bm a}^{(k)}_E$ for all bath modes $\mu>n$.

\subsubsection{Heisenberg evolution induced by split-step propagation}

Because bulk turbulence is polarization-independent, the single-photon propagator
on the one-particle space factorizes as
\begin{equation}
\mathcal{U}^{(1)}_{k+1:k}=\mathcal{U}^{(\perp)}_{k+1:k}\otimes I_{\mathrm{pol}},
\end{equation}
where $\mathcal{U}^{(\perp)}_{k+1:k}=\mathcal{P}_{\Delta z}\circ\mathcal{S}_{\phi_k}$
acts on transverse functions and $I_{\mathrm{pol}}$ acts on polarization.

Let $\hat U_{k+1:k}$ denote the corresponding passive unitary acting on the full
bosonic Fock space. It satisfies the standard linear-optics Heisenberg relation
\begin{equation}
\hat U_{k+1:k}^\dagger\, \hat a^{(k+1)}_{\mu,p}\,\hat U_{k+1:k}
=
\sum_{\mu'\ge 1}
t^{(k)}_{\mu\leftarrow \mu'}\, \hat a^{(k)}_{\mu',p},
\label{eq:heisenberg_spatial}
\end{equation}
where $t^{(k)}_{\mu\leftarrow\mu'}$ are exactly the overlap coefficients in
\eqref{eq:t_coeff_def}--\eqref{eq:t_coeff_explicit}. Polarization is unchanged in
bulk propagation, hence the absence of mixing between $p$ and $p'$ in
\eqref{eq:heisenberg_spatial}.

\subsubsection{System block and bath block}

Restricting \eqref{eq:heisenberg_spatial} to the kept modes $\mu\le n$ yields
\begin{equation}
\hat a^{(k+1)}_{m,p}
=
\sum_{m'=1}^n t^{(k)}_{m\leftarrow m'} \hat a^{(k)}_{m',p}
+
\sum_{\mu'>n} t^{(k)}_{m\leftarrow \mu'} \hat a^{(k)}_{\mu',p},
\qquad m\le n.
\label{eq:kept_heisenberg}
\end{equation}
This can be written in compact block form as
\begin{equation}
\hat{\bm a}^{(k+1)}_S
=
T_k\,\hat{\bm a}^{(k)}_S
+
R_k\,\hat{\bm a}^{(k)}_E,
\label{eq:system_block_relation}
\end{equation}
where, in the polarization-independent case,
\begin{equation}
T_k = T_k^{\perp}\otimes I_2,
\qquad
R_k = R_k^{\perp}\otimes I_2.
\label{eq:Tk_tensor}
\end{equation}

More generally, if polarization mixing is introduced by receiver optics or a
per-rail polarization transformation represented by Jones matrices
$J_{k,(m)}\in\mathbb{C}^{2\times 2}$ acting on rail $m$, then the system matrix
elements are
\begin{equation}
(T_k)_{(m,p),(m',p')}
=
J_{k,(m)}(p,p')\,
\left\langle
w_m^{(k+1)},\ \mathcal{U}_{k+1:k} w_{m'}^{(k)}
\right\rangle,
\label{eq:Tk_general}
\end{equation}
which reduces to \eqref{eq:Tk_tensor} when $J_{k,(m)}=I_2$.

\subsubsection{Contraction identity}

Unitarity of the underlying mode transformation implies
\begin{equation}
T_k T_k^\dagger + R_k R_k^\dagger = I_{2n},
\label{eq:contraction_2n}
\end{equation}
which is the polarization-resolved analogue of \eqref{eq:contraction_spatial}. It
encodes that leakage to bath modes renders the effective system evolution
non-unitary, thereby necessitating a Kraus (open-system) description at the
logical level.

The subsequent subsections lift \eqref{eq:system_block_relation} to the logical
$n$-qubit encoding and derive explicit Kraus operators for both distinguishable
and indistinguishable photon regimes.

\subsection{Quantization, System--Bath Factorization, and Reduced Dynamics}

We now lift the polarization-resolved single-particle transformation
\eqref{eq:system_block_relation} to the quantum many-body level and derive the
reduced system dynamics by tracing out the bath modes.

\subsubsection{Bosonic Fock space and factorization}

Let $\mathcal{H}^{(k)}_{\mathrm{kept}} \simeq \mathbb{C}^{2n}$ denote the
single-particle kept subspace (spatial rails $\times$ polarization) at plane
$z_k$, and let $\mathcal{H}^{(k)}_{\mathrm{bath}}$ denote its orthogonal
complement. The full single-particle space decomposes as
\begin{equation}
\mathcal{H} = \mathcal{H}^{(k)}_{\mathrm{kept}} \oplus \mathcal{H}^{(k)}_{\mathrm{bath}}.
\end{equation}
Accordingly, the bosonic Fock space factorizes as
\begin{equation}
\mathcal{F}(\mathcal{H})
\simeq
\mathcal{F}(\mathcal{H}^{(k)}_{\mathrm{kept}})
\otimes
\mathcal{F}(\mathcal{H}^{(k)}_{\mathrm{bath}}).
\label{eq:fock_factorization}
\end{equation}

The passive unitary $\hat U_{k+1:k}$ induced by split-step propagation acts on
$\mathcal{F}(\mathcal{H})$ and implements the Heisenberg transformation
\eqref{eq:system_block_relation}. We assume that the bath modes defined at plane
$z_k$ are initially in the vacuum state, which is consistent with performing the
system--bath split by modal decomposition at each plane.

\subsubsection{Range-step reduced map}

Let $\rho_S(z_k)$ be a density operator on
$\mathcal{F}(\mathcal{H}^{(k)}_{\mathrm{kept}})$. The reduced system state at
plane $z_{k+1}$ is
\begin{equation}
\rho_S(z_{k+1})
=
\Tr_E\!\left[
\hat U_{k+1:k}
\bigl(\rho_S(z_k)\otimes\ket{\mathrm{vac}}\bra{\mathrm{vac}}_E\bigr)
\hat U_{k+1:k}^\dagger
\right].
\label{eq:reduced_map_fock}
\end{equation}
Equation \eqref{eq:reduced_map_fock} defines a completely positive trace-preserving
(CPTP) map on the kept-mode Fock space. The remaining task is to restrict this map
to the logical polarization encoding and obtain explicit Kraus operators.

\subsection{Logical Encoding and Erasure-Augmented System Space}

\subsubsection{Logical subspace}

At the transmitter plane, the logical encoding consists of exactly one photon in
each spatial rail, with polarization encoding the qubit value. The logical
subspace is
\begin{equation}
\mathcal{H}_{\mathrm{log}}
=
\mathrm{span}\Bigl\{
\hat a_{1,p_1}^\dagger\cdots \hat a_{n,p_n}^\dagger\ket{\mathrm{vac}}
:\; p_m\in\{H,V\}
\Bigr\}
\simeq (\mathbb{C}^2)^{\otimes n}.
\label{eq:logical_subspace}
\end{equation}

Because \eqref{eq:system_block_relation} includes leakage into bath modes and
mixing among spatial rails, $\mathcal{H}_{\mathrm{log}}$ is not invariant under
the reduced evolution \eqref{eq:reduced_map_fock}.

\subsubsection{Erasure-augmented encoding}

To obtain a fixed-dimensional CPTP description at the logical level, we extend
each rail with an explicit erasure flag. For each spatial rail $m$, define the
local system space
\begin{equation}
\mathcal{H}_m = \mathbb{C}^2 \oplus \mathrm{span}\{\ket{\emptyset}\},
\end{equation}
where $\ket{\emptyset}$ represents the event that the photon initially associated
with rail $m$ has leaked from the kept subspace or participates in a non-logical
occupation pattern\footnote{\textcolor{black}{$|\varnothing\rangle$ represents a flagged erasure event, corresponding to any outcome in which the photon initially associated with rail $m$ cannot be assigned to the logical mode at the receiver. Physically, this includes scattering into spatial modes outside the finite set collected by the receiver (e.g., power leaving the aperture or failing mode projection), as well as non-logical occupation patterns such as multiple photons exiting in the same spatial rail that cannot be mode resolved.
}}.

The full erasure-augmented system space is
\begin{equation}
\mathcal{H}_{\mathrm{sys}} = \bigotimes_{m=1}^n \mathcal{H}_m,
\qquad \dim \mathcal{H}_{\mathrm{sys}} = 3^n.
\label{eq:erasure_space}
\end{equation}

Let $\Pi_{\mathrm{sys}}$ denote the isometry embedding
$\mathcal{H}_{\mathrm{sys}}$ into the kept-mode Fock space, which maps
$\ket{H}_m,\ket{V}_m$ to single-photon states in modes $(m,H)$ and $(m,V)$,
respectively, and maps $\ket{\emptyset}_m$ to the vacuum in rail $m$. This
embedding allows all non-logical outcomes of the physical evolution to be mapped
to orthogonal erasure states.

\subsection{Distinguishable-Photon Regime: Explicit Kraus Operators}

We first consider the regime in which the $n$ photons are distinguishable by an
additional degree of freedom (e.g., time bin or frequency tag), so that
multi-photon interference is suppressed.

\subsubsection{Single-rail polarization map}

In this regime, each photon can be treated independently. For a given rail $m$,
extract from the system matrix $T_k$ the $2\times 2$ polarization block
\begin{equation}
B_{m,k}
=
\begin{pmatrix}
(T_k)_{(m,H),(m,H)} & (T_k)_{(m,H),(m,V)} \\
(T_k)_{(m,V),(m,H)} & (T_k)_{(m,V),(m,V)}
\end{pmatrix}.
\label{eq:Bm_def}
\end{equation}
From the contraction identity \eqref{eq:contraction_2n}, it follows that
\begin{equation}
B_{m,k}^\dagger B_{m,k} \le I_2.
\end{equation}
Define the complementary positive operator
\begin{equation}
C_{m,k} = \sqrt{I_2 - B_{m,k}^\dagger B_{m,k}}.
\label{eq:Cm_def}
\end{equation}
Under the wrong-port-is-not-an-error convention,  we interpret $B_{m,k}$ as the effective single-rail\footnote{We use the term \emph{rail} to denote a logical spatial-mode channel associated with a selected orthogonal spatial degree of freedom and its polarization degree of freedom. Each rail carries exactly one photon at the transmitter, with polarization encoding a qubit. Depending on the physical realization, different rails may be implemented by orthogonal co-propagating modes, by beams with partial spatial separation, or more generally by transmitter–receiver mode pairs chosen to define the kept system subspace. Atmospheric turbulence can then induce both coupling between rails and leakage into unobserved spatial degrees of freedom. When the rails experience a common or partially common turbulent medium, these effects appear as correlated fading, inter-rail crosstalk, and spillover; when they are sufficiently separated relative to the turbulence coherence scale, the channel responses become progressively less correlated.} 
 polarization block after rail relabeling. 
\subsubsection{Per-rail Kraus operators}

Define Kraus operators acting from $\mathbb{C}^2$ to $\mathbb{C}^2\oplus\ket{\emptyset}$
by
\begin{align}
K^{(m)}_{0,k}
&=
\begin{pmatrix}
B_{m,k} \\
0
\end{pmatrix},
\\
K^{(m)}_{1,k}
&=
\begin{pmatrix}
0 \\
(C_{m,k})_{1,:}
\end{pmatrix},
\qquad
K^{(m)}_{2,k}
=
\begin{pmatrix}
0 \\
(C_{m,k})_{2,:}
\end{pmatrix}.
\label{eq:single_rail_kraus}
\end{align}
These satisfy the completeness relation
\begin{equation}
\sum_{\alpha=0}^2 K^{(m)\dagger}_{\alpha,k} K^{(m)}_{\alpha,k} = I_2.
\end{equation}

\subsubsection{Multi-qubit Kraus representation}

The $n$-qubit reduced map on $\mathcal{H}_{\mathrm{sys}}$ for one range step is
\begin{equation}
\rho_S(z_{k+1})
=
\sum_{\bm{\alpha}\in\{0,1,2\}^n}
\left(
\bigotimes_{m=1}^n K^{(m)}_{\alpha_m,k}
\right)
\rho_S(z_k)
\left(
\bigotimes_{m=1}^n K^{(m)}_{\alpha_m,k}
\right)^\dagger.
\label{eq:distinguishable_kraus}
\end{equation}
Although the Kraus operators factorize, the channel is generally correlated
across qubits because all $B_{m,k}$ depend on the same turbulence realization via
the overlap integrals \eqref{eq:t_coeff_explicit}.

\subsection{Indistinguishable-Photon Regime: Emergence of Permanents}

Although the photons are initially prepared in orthogonal spatial modes and are therefore distinguishable at the transmitter, this distinguishability is not necessarily preserved at the receiver. Atmospheric turbulence induces intermodal mixing, and the receiver performs a projection onto a finite set of spatial modes. As a result, the information about the photon’s original spatial-mode label may be partially or completely erased. When no additional degree of freedom (such as time, frequency, or polarization) encodes which-photon information, the photons become effectively indistinguishable at detection. In this regime, the output statistics are governed by bosonic many-body interference, leading to transition amplitudes expressed in terms of matrix permanents.

We now treat the regime in which the $n$ photons are indistinguishable bosons
occupying the same temporal and spectral mode.

\subsubsection{Bosonic lifting of the single-particle map}

The single-particle Heisenberg transformation \eqref{eq:system_block_relation}
implies the creation-operator mapping
\begin{equation}
\hat a_{j,p_j}^\dagger
\mapsto
\sum_{(m,q)\in S}
(T_k)_{(m,q),(j,p_j)}\,\hat a_{m,q}^\dagger
+
\sum_{\ell\in E}
(R_k)_{\ell,(j,p_j)}\,\hat b_\ell^\dagger,
\label{eq:creation_map_full}
\end{equation}
where $S$ indexes system output modes and $E$ indexes bath modes.

\subsubsection{Derivation of the permanent}

Let the logical input state be
\begin{equation}
\ket{p}
=
\hat a_{1,p_1}^\dagger\cdots \hat a_{n,p_n}^\dagger\ket{\mathrm{vac}}.
\end{equation}
Expanding the product of \eqref{eq:creation_map_full} for $j=1,\dots,n$ yields a
coherent sum over all assignments of input photons to output modes. If we
postselect on logical output states
\begin{equation}
\ket{q}
=
\hat a_{1,q_1}^\dagger\cdots \hat a_{n,q_n}^\dagger\ket{\mathrm{vac}},
\end{equation}
with exactly one photon per rail, then only terms corresponding to permutations
$\pi\in S_n$ contribute. Using bosonic commutation relations, all such terms add
with the same sign, giving
\begin{equation}
\braket{q|U_k|p}
=
\sum_{\pi\in S_n}
\prod_{i=1}^n
(T_k)^*_{(i,q_i),(\pi(i),p_{\pi(i)})}
=
\mathrm{perm}\!\bigl(M_k(q|p)\bigr),
\label{eq:permanent_result}
\end{equation}
where the $n\times n$ matrix $M_k(q|p)$ has entries
\begin{equation}
[M_k(q|p)]_{ij} = (T_k)^*_{(i,q_i),(j,p_j)}.
\end{equation}

\subsubsection{Logical Kraus operators}

Including all bath outcomes and mapping non-logical events to erasure states via
$\Pi_{\mathrm{sys}}$, the reduced evolution on $\mathcal{H}_{\mathrm{sys}}$ admits
a Kraus representation
\begin{equation}
\rho_S(z_{k+1})
=
\sum_\beta
K_{\beta,k}\,\rho_S(z_k)\,K_{\beta,k}^\dagger,
\end{equation}
with
\begin{equation}
K_{\beta,k}
=
\Pi_{\mathrm{sys}}^\dagger
\bra{e_\beta}
\hat U_{k+1:k}
\bigl(\Pi_{\mathrm{sys}}\otimes\ket{\mathrm{vac}}_E\bigr).
\label{eq:indist_kraus}
\end{equation}
In the indistinguishable-photon regime, the logical blocks of the Kraus operators
contain coherent sums of permanents whose entries are given explicitly by the
overlap integrals \eqref{eq:t_coeff_explicit}, thereby completing the derivation
from turbulence parameters to logical $n$-qubit noise. We should highlight that inter-system crosstalk is not fundamentally noise if the spatial-mode coupling were deterministic, known, and lossless. It could in principle be compensated by appropriate pre- or post-processing. In free-space propagation through atmospheric turbulence, however, the mode coupling is induced by a random, three-dimensional refractive-index field and is only partially observable at the receiver due to finite mode collection. Turbulence renders the intermodal mixing effectively unknown and induces scattering into spatial modes outside the observed set, leading to irreversible loss. While adaptive optics can mitigate phase distortions and reduce modal distortion, it cannot fully suppress scintillation or recover energy that has leaked outside the collected mode set. As a result, the combination of partial observability (finite mode projection) and stochastic propagation leads, after tracing over unobserved degrees of freedom and averaging over turbulence realizations, to intrinsically non-unitary reduced dynamics on the logical polarization degrees of freedom, which are naturally described as noise.

\subsection{Full-Range Evolution: Composition of Range-Step Maps}

The previous subsections derived the reduced dynamics for a \emph{single} range
step $[z_k,z_{k+1}]$. We now construct the \emph{full end-to-end Quantum MIMO
channel} describing propagation from the transmitter plane $z=0$ to the receiver
plane $z=L$ by explicit composition of these range-step maps.

\subsubsection{Concatenation at the level of single-particle propagation}

At the single-photon level, the transverse propagation operator from $z=0$ to
$z=L$ is
\begin{equation}
\mathcal{U}_{L:0}
=
\mathcal{U}_{K:K-1}\circ \mathcal{U}_{K-1:K-2}\circ \cdots \circ \mathcal{U}_{1:0},
\label{eq:U_full}
\end{equation}
where each factor $\mathcal{U}_{k+1:k}=\mathcal{P}_{\Delta z}\circ\mathcal{S}_{\phi_k}$
depends explicitly on the correlated phase screen $\phi_k(\mathbf{r})$ generated
from the same three-dimensional turbulence realization.

Let $\{w_m^{(0)}\}$ denote the launched spatial modes at $z=0$, and
$\{w_m^{(K)}\}$ denote the detected spatial modes at $z=L$. The full-range spatial
overlap coefficients are
\begin{equation}
t^{(L)}_{m\leftarrow m'}
=
\left\langle
w_m^{(K)},\ \mathcal{U}_{L:0} w_{m'}^{(0)}
\right\rangle.
\label{eq:t_full}
\end{equation}
These coefficients encode \emph{all} cumulative diffraction, turbulence-induced
scattering, and mode mismatch along the path.

\subsubsection{Full-range system matrix}

At the operator level, repeated application of the Heisenberg relations
\eqref{eq:system_block_relation} yields
\begin{equation}
\hat{\bm a}^{(K)}_S
=
T_{L:0}\,\hat{\bm a}^{(0)}_S
+
\sum_{j=0}^{K-1}
\left(
T_{K:j+1} R_j
\right)
\hat{\bm a}^{(j)}_E,
\label{eq:system_block_full}
\end{equation}
where
\begin{equation}
T_{L:0} := T_{K-1} T_{K-2}\cdots T_0
\label{eq:T_full}
\end{equation}
is the \emph{full-range system matrix}, and
\begin{equation}
T_{K:j+1} := T_{K-1}T_{K-2}\cdots T_{j+1}
\end{equation}
with the convention $T_{K:K}=I$.

Equation \eqref{eq:system_block_full} shows explicitly that the final system
operators depend not only on the initial system operators but also on bath
operators injected at all intermediate planes. This is the microscopic origin of
loss, erasure, and non-unitarity in the reduced logical channel.

\subsubsection{Full-range reduced map on the kept-mode Fock space}

Let $\rho_S(0)$ be the initial system density operator on
$\mathcal{F}(\mathcal{H}^{(0)}_{\mathrm{kept}})$. The reduced system state at the
receiver plane is
\begin{equation}
\rho_S(L)
=
\Tr_E\!\left[
\hat U_{L:0}
\bigl(\rho_S(0)\otimes\ket{\mathrm{vac}}\bra{\mathrm{vac}}_E\bigr)
\hat U_{L:0}^\dagger
\right],
\label{eq:full_reduced_map}
\end{equation}
where $\hat U_{L:0}=\hat U_{K:K-1}\cdots\hat U_{1:0}$ is the full Fock-space unitary.

Equation \eqref{eq:full_reduced_map} defines a CPTP map $\Lambda_{L:0}$ on the
kept-mode Fock space. Importantly, because the same turbulence realization
generates all phase screens $\{\phi_k\}$, the sequence of maps
$\{\Lambda_k\}_{k=0}^{K-1}$ is \emph{correlated} and the reduced evolution is
generally non-Markovian with respect to the range parameter.

\subsubsection{Full-range logical channel: distinguishable photons}

In the distinguishable-photon regime, each range-step map admits a Kraus
representation \eqref{eq:distinguishable_kraus}. The full-range logical channel is
their ordered composition,
\begin{equation}
\Lambda^{(\mathrm{dist})}_{L:0}
=
\Lambda^{(\mathrm{dist})}_{K-1}\circ
\Lambda^{(\mathrm{dist})}_{K-2}\circ \cdots \circ
\Lambda^{(\mathrm{dist})}_{0}.
\label{eq:full_dist_channel}
\end{equation}
Equivalently, $\Lambda^{(\mathrm{dist})}_{L:0}$ admits a Kraus representation whose
Kraus operators are all products
\begin{equation}
\tilde K_{\bm{\alpha}}
=
\prod_{k=0}^{K-1}
\left(
\bigotimes_{m=1}^n K^{(m)}_{\alpha_{m,k},k}
\right),
\label{eq:full_dist_kraus}
\end{equation}
with multi-indices $\bm{\alpha}=\{\alpha_{m,k}\}_{m,k}$. All dependence on the
turbulence parameters $(C_n^2,L_0,l_0)$ enters through the product
$T_{L:0}=T_{K-1}\cdots T_0$.

\subsubsection{Full-range logical channel: indistinguishable photons}

In the indistinguishable-photon regime, the relevant object is the full-range
system matrix $T_{L:0}$. Repeating the derivation leading to
\eqref{eq:permanent_result}, but with $T_k$ replaced by $T_{L:0}$, yields the
end-to-end logical transition amplitudes
\begin{equation}
\braket{q|U_{L:0}|p}
=
\mathrm{perm}\!\left(M_{L:0}(q|p)\right),
\label{eq:full_perm}
\end{equation}
where
\begin{equation}
[M_{L:0}(q|p)]_{ij}
=
(T_{L:0})^*_{(i,q_i),(j,p_j)}.
\end{equation}
Thus, the entire propagation through a turbulent path of length $L$ enters the
logical channel \emph{only} through the full-range overlap matrix $T_{L:0}$,
itself determined by the correlated phase screens generated from the turbulence
parameters.

It is important to highlight at this stage that this resulting logical channel is naturally described as a correlated erasure channel, as we will show subsequentely, where turbulence-induced leakage of spatial modes leads to erasure events that are correlated across rails. Conditioned on survival, the polarization degree of freedom is preserved by the atmospheric channel and may additionally experience correlated unitary perturbations due to terminal-induced drifts. Only after conditioning on non-erasure and applying Pauli twirling can the polarization dynamics be approximated by a correlated Pauli channel. 

\subsection{Relation to Classical Fading and Outage Models}
\label{sec:fading_relation}

Classical free-space optical (FSO) and wireless channels are commonly modeled as fading channels, in which the received signal is multiplied by a random complex gain determined by the propagation medium. In the single-input single-output (SISO) case, this is written as
\begin{equation}
y = h x + n,
\end{equation}
where the random coefficient $h$ captures turbulence-induced amplitude and phase fluctuations, and deep fades $|h|^2 \ll 1$ give rise to outage events. For multiple-input multiple-output (MIMO) systems, this generalizes to a random channel matrix $\mathbf{H}$ whose entries may be correlated due to shared propagation through the atmosphere.

The Quantum MIMO free-space optical channel derived in this work constitutes a direct physical generalization of such fading models. Rather than postulating a phenomenological fading coefficient or matrix, the effective channel is determined by the random spatial-mode overlap matrix
\begin{equation}
T_{\perp}(\omega) = \left[ t_{r \leftarrow m}(\omega) \right]_{r,m}
= \left[ \langle w_r, \mathcal{U}_{L:0} u_m \rangle \right]_{r,m},
\end{equation}
which arises from wave-optical propagation through a three-dimensional turbulent refractive-index field. This matrix plays the role of a \emph{fading matrix}, with randomness inherited directly from the underlying turbulence realization $\omega$.

A key distinction from classical fading models is that $T_{\perp}(\omega)$ is generally \emph{subunitary}, reflecting irreversible scattering of optical power into spatial modes outside the receiver's kept subspace. The per-rail survival probability
\begin{equation}
p_{\mathrm{keep},m}(\omega) = \sum_{r=1}^{n} \left| t_{r \leftarrow m}(\omega) \right|^2
\end{equation}
therefore serves as the quantum analogue of an instantaneous channel power gain. Events in which $p_{\mathrm{keep},m}(\omega)$ is small correspond to \emph{deep fades}, in which the photon initially associated with rail $m$ is unlikely to be detected within the receiver mode set. At the logical level, such deep fades naturally manifest as \emph{erasure events}, with erasure probability
\begin{equation}
\epsilon_m(\omega) = 1 - p_{\mathrm{keep},m}(\omega).
\end{equation}

Because all spatial modes propagate through the same turbulent medium, the random variables $\{ \epsilon_m(\omega) \}_{m=1}^{n}$ are generally correlated across rails. This directly parallels correlated fading in classical MIMO channels, where different spatial streams experience statistically dependent channel gains. In the present quantum setting, these correlations appear as \emph{correlated erasures} in the logical $n$-qubit channel, as formalized in Section. \ref{sec:reduction_correlated_erasure}.

Longitudinal correlations in the refractive-index fluctuations further induce correlations across propagation range, giving rise to non-Markovian channel evolution. This corresponds to slow- or block-fading regimes in classical communications, but here emerges from a microscopic three-dimensional turbulence model rather than an imposed statistical assumption.

In summary, the Quantum MIMO channel derived in this work can be viewed as a physically grounded fading channel in which (i) random mode-overlap matrices replace phenomenological fading coefficients, (ii) turbulence-induced leakage corresponds to deep fades that induce erasures, and (iii) shared propagation through a single random medium leads to correlated fading across spatial channels.

\section{Reduction to a Correlated $n$-Qubit Erasure Channel}
\label{sec:reduction_correlated_erasure}

This section formalizes how the quantum MIMO FSO channel derived in
Sec.~IV reduces, at the logical level, to a \emph{correlated $n$-qubit erasure
channel} under physically standard assumptions. The reduction makes explicit
(i) the origin of correlations across spatial rails and (ii) the distinction
between the \emph{extended} (flagged, per-rail) erasure map and the
\emph{non-extended} (coarse-grained, block) erasure map.

\subsection{From Turbulence-Induced Mode Mixing to Per-Rail Survival}
\label{subsec:per_rail_survival}

Fix a propagation distance $L$ and consider $n$ transmitted spatial modes
$\{u_m\}_{m=1}^n$ and $n$ collected (``kept'') receiver modes
$\{w_r\}_{r=1}^n$ (e.g., LG modes). For each turbulence realization $\omega$,
the split-step propagation with random phase screens produces the received field
$\psi_m(\cdot;L,\omega)$ and the kept-set overlap (crosstalk) matrix
\begin{equation}
t_{r\leftarrow m}(\omega)
=
\langle w_r,\psi_m(\cdot;L,\omega)\rangle,
\qquad
T^\perp(\omega)=[t_{r\leftarrow m}(\omega)]_{r,m=1}^n.
\label{eq:Tperp_omega}
\end{equation}
Throughout this reduction we adopt the operational convention used in the main
text: \emph{arrival in a wrong kept port is not an error}. Thus, the relevant
logical impairment due to spatial turbulence is \emph{leakage outside the kept
subspace}.

Accordingly, define the \emph{per-rail kept probability} for input rail $m$ as
the total probability mass remaining within the kept subspace:
\begin{align}
&p_{\mathrm{keep},m}(\omega)
=
\sum_{r=1}^n |t_{r\leftarrow m}(\omega)|^2,
\nonumber\\
&\epsilon_m(\omega)= p_{\mathrm{erase},m}(\omega)=1-p_{\mathrm{keep},m}(\omega).
\label{eq:pkeep_per_rail}
\end{align}
The random vector $\boldsymbol{\epsilon}(\omega)=(\epsilon_1(\omega),\dots,\epsilon_n(\omega))$
encodes the realization-dependent erasure propensities across rails. Since the
same turbulent medium acts jointly on all input modes, the components of
$\boldsymbol{\epsilon}(\omega)$ are generally \emph{statistically dependent}
across $m$; this dependence is the fundamental origin of \emph{correlated
erasures} at the logical level.

We remark that, for isotropic atmospheric turbulence modeled as a scalar refractive-index
fluctuation, the spatial coupling is polarization-independent, and the overall
single-photon map on the kept set takes the form
$T(\omega)=T^\perp(\omega)\otimes I_2$. Hence polarization is preserved
\emph{conditioned on survival}, and turbulence-induced logical errors manifest
primarily as erasures. Polarization noise can be incorporated separately via
mode-dependent or rail-dependent Jones matrices (cf. Sec.~IV), in which case the
logical channel becomes ``correlated erasure + correlated polarization noise''
rather than pure erasure.

\subsection{Extended (Flagged) Erasure Map}
\label{subsec:extended_erasure}

To represent photon loss on individual rails, we use the standard \emph{extended}
erasure embedding. For each rail $m$, the logical polarization qubit Hilbert
space $\mathcal{H}_m\simeq\mathbb{C}^2$ is enlarged to
\begin{equation}
\tilde{\mathcal{H}}_m \;=\; \mathcal{H}_m \oplus \mathrm{span}\{\ket{e}\},
\label{eq:extended_space}
\end{equation}
where $\ket{e}$ is an orthogonal erasure flag state. The single-rail erasure
channel with erasure probability $\epsilon$ is
\begin{equation}
\mathcal{E}_{\epsilon}(\rho)=(1-\epsilon)\rho+\epsilon\,\ket{e}\!\bra{e}.
\label{eq:single_rail_erasure_channel}
\end{equation}

We now restrict attention to the distinguishable-photon regime (i.e., no multi-photon interference), so that conditioned on a fixed turbulence realization $\omega$ the $n$ photons can be treated independently across rails. Indeed, if conditioned on a fixed turbulence realization $\omega$, each photon evolves
independently through the linear optical medium and is erased with probability
$\epsilon_m(\omega)$ on rail $m$, independently across $m$. This is the natural multi-photon extension of the
single-photon leakage model when photons are generated independently and the
dominant impairment is coupling to unobserved spatial modes. Under scalar
turbulence, polarization does not affect the leakage probability, so the
conditional polarization map is the identity on non-erased rails. As a matter of fact, the \emph{conditional} logical channel
given $\omega$ factorizes:
\begin{equation}
\Lambda_{\omega}(\rho)
=
\bigotimes_{m=1}^n \mathcal{E}_{\epsilon_m(\omega)}(\rho),
\label{eq:Lambda_omega_product}
\end{equation}
where $\rho$ acts on $(\mathbb{C}^2)^{\otimes n}$ and the output acts on
$(\mathbb{C}^2\oplus\ket{e})^{\otimes n}$. For each erasure pattern
$\mathbf{s}=(s_1,\dots,s_n)\in\{0,1\}^n$ (with $s_m=1$ indicating that rail $m$ is
erased), the conditional probability of $\mathbf{s}$ is
\begin{equation}
p(\mathbf{s}\mid \omega)
=
\prod_{m=1}^n
\big(\epsilon_m(\omega)\big)^{s_m}
\big(1-\epsilon_m(\omega)\big)^{1-s_m}.
\label{eq:pattern_prob_conditional}
\end{equation}

The physically relevant logical channel is obtained by averaging over turbulence
realizations:
\begin{equation}
\Lambda(\rho)=\mathbb{E}_{\omega}\!\left[\Lambda_{\omega}(\rho)\right].
\label{eq:Lambda_ensemble_def}
\end{equation}
Averaging \eqref{eq:pattern_prob_conditional} over $\omega$ yields the joint
erasure-pattern law
\begin{equation}
p(\mathbf{s})
=
\mathbb{E}_{\omega}\!\left[
\prod_{m=1}^n
\big(\epsilon_m(\omega)\big)^{s_m}
\big(1-\epsilon_m(\omega)\big)^{1-s_m}
\right].
\label{eq:pattern_prob_unconditional}
\end{equation}
In general, the random variables $\{\epsilon_m(\omega)\}_{m=1}^n$ are correlated
across $m$, so \eqref{eq:pattern_prob_unconditional} does not factorize into a
product of marginals. Hence the ensemble-averaged logical channel is a
\emph{correlated $n$-qubit erasure channel}.

\begin{proposition}[Correlated erasure representation]
\label{prop:correlated_erasure_channel}
Under scalar turbulence and the distinguishable-photon (no multi-photon interference) regime described above, the ensemble
logical channel admits the representation
\begin{equation}
\Lambda(\rho)
=
\sum_{\mathbf{s}\in\{0,1\}^n}
p(\mathbf{s})\;
\Big(\rho_{\bar{\mathbf{s}}}\otimes
\ket{e\cdots e}_{\mathbf{s}}\!\bra{e\cdots e}\Big),
\label{eq:correlated_erasure_channel}
\end{equation}
where $\rho_{\bar{\mathbf{s}}}$ denotes the reduced state on the non-erased rails
(i.e., the identity action on polarization on those rails) and the erased rails
are replaced by $\ket{e}$.
\end{proposition}

\begin{proof}
Conditioned on $\omega$, the product form \eqref{eq:Lambda_omega_product} is a
mixture over patterns $\mathbf{s}$ with weights
$p(\mathbf{s}\mid\omega)$ given by \eqref{eq:pattern_prob_conditional}, where the
map acts as the identity on non-erased rails and outputs $\ket{e}$ on erased
rails. Taking the expectation over $\omega$ yields the same mixture with weights
$p(\mathbf{s})=\mathbb{E}_{\omega}[p(\mathbf{s}\mid\omega)]$, which is exactly
\eqref{eq:pattern_prob_unconditional}, giving \eqref{eq:correlated_erasure_channel}.
\end{proof}

\subsection{Non-Extended (Coarse-Grained) Erasure Map}
\label{subsec:nonextended_erasure}

In some protocols one does not retain the full erasure pattern $\mathbf{s}$ but
only a \emph{block-level} success/failure flag. Define the block success event
as ``no rail is erased,'' i.e. $\mathbf{s}=\mathbf{0}$. Let
\begin{equation}
p_{\mathrm{succ}}= p(\mathbf{0})
=\mathbb{E}_{\omega}\!\left[\prod_{m=1}^n\big(1-\epsilon_m(\omega)\big)\right].
\label{eq:psucc_def}
\end{equation}
The corresponding coarse-grained (non-extended) erasure map is
\begin{equation}
\Lambda_{\mathrm{coarse}}(\rho)
=
p_{\mathrm{succ}}\,\rho
+\big(1-p_{\mathrm{succ}}\big)\ket{E}\!\bra{E},
\label{eq:coarse_erasure_channel}
\end{equation}
where $\ket{E}$ is a single global erasure flag orthogonal to the logical output
space. Notably, \eqref{eq:coarse_erasure_channel} discards information about
\emph{which rails} were erased and therefore cannot capture partial-erasure
events (e.g., losing one photon but retaining others), which are naturally
represented by the extended model \eqref{eq:correlated_erasure_channel}.

\subsection{Estimating the Correlated Erasure Law from Optical Simulations}
\label{subsec:estimating_ps}

The reduction \eqref{eq:pattern_prob_unconditional}--\eqref{eq:correlated_erasure_channel}
provides a direct computational pathway from optical simulation outputs to the
correlated erasure channel parameters. Given $N_{\mathrm{MC}}$ turbulence
realizations $\{\omega_j\}_{j=1}^{N_{\mathrm{MC}}}$, compute $T^\perp(\omega_j)$
and $\epsilon_m(\omega_j)$ via \eqref{eq:pkeep_per_rail}. Then the joint erasure
law can be estimated for any pattern $\mathbf{s}\in\{0,1\}^n$ by
\begin{equation}
\hat p(\mathbf{s})
=
\frac{1}{N_{\mathrm{MC}}}
\sum_{j=1}^{N_{\mathrm{MC}}}
\prod_{m=1}^n
\big(\epsilon_m(\omega_j)\big)^{s_m}
\big(1-\epsilon_m(\omega_j)\big)^{1-s_m}.
\label{eq:p_hat_pattern}
\end{equation}
In particular, the block success probability $p_{\mathrm{succ}}$ in
\eqref{eq:psucc_def} is estimated by
\begin{equation}
\hat p_{\mathrm{succ}}
=
\frac{1}{N_{\mathrm{MC}}}
\sum_{j=1}^{N_{\mathrm{MC}}}
\prod_{m=1}^n\big(1-\epsilon_m(\omega_j)\big).
\label{eq:p_succ_hat}
\end{equation}

To quantify correlated erasures, one may form Bernoulli erasure indicators
$L_m\in\{0,1\}$ per rail by sampling $L_m\sim\mathrm{Bernoulli}(\epsilon_m(\omega))$
for each realization. Pairwise dependence can then be summarized by
$\mathrm{Corr}(L_m,L_{m'})$. These
statistics are directly computable from the Monte Carlo ensemble and provide
compact KPIs for characterizing the correlated erasure structure implied by the
optical channel.

We highlight that, the reduction to a pure correlated erasure channel is exact under:
(i) the wrong-port-is-not-an-error convention (spatial mixing within the kept set
is treated as harmless or is corrected by mode sorting) and (ii) scalar turbulence
($T^\perp\otimes I_2$). If polarization-dependent effects are introduced (e.g., mode-dependent Jones
matrices or space--polarization coupling in the receiver), then the conditional
non-erasure map on polarization is no longer the identity, and the logical
channel becomes a composition of correlated erasure with correlated polarization
noise.


\subsection{Simulation Parameters and Methods}
\label{subsec:simulation_parameters}

\begin{figure}[t!]
    \centering
    \includegraphics[width = 0.475\textwidth]{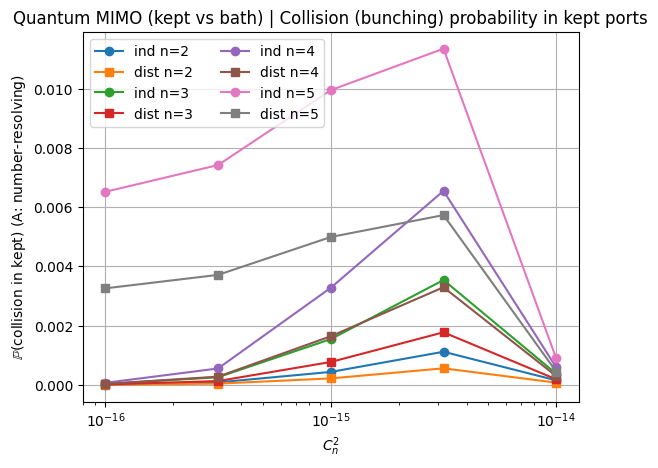}
    \caption{%
    \textbf{Multi-photon collision (bosonic bunching) probability in the kept ports.}
    Probability that two or more photons exit in the same kept spatial mode,
    conditioned on number-resolving detection, shown as a function of turbulence
    strength $C_n^2$ for multiplexing orders $n=2$–$5$.
    Results are shown for both indistinguishable (bosonic) and distinguishable
    photons.
    Bosonic enhancement of collision events is evident at intermediate turbulence,
    while strong turbulence suppresses collisions due to dominant erasure.}

\color{black}
    \label{fig:QMIMO1}
\end{figure}

\begin{figure}[t!]
    \centering
    \includegraphics[width = 0.475\textwidth]{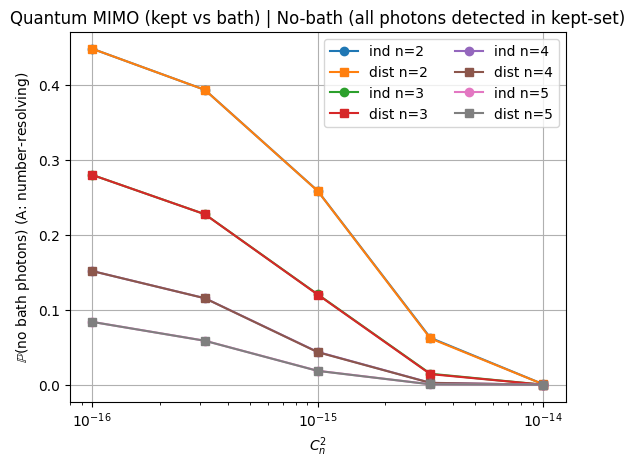}
    \caption{%
    \textbf{No-bath postselected detection probability.}
    Probability that all $n$ photons are detected within the kept receiver mode
    set (i.e., no erasure), assuming perfect postselection and no access to bath
    modes.
    Results are plotted versus turbulence strength $C_n^2$ for
    multiplexing orders $n=2$–$5$.
    Increasing turbulence and higher spatial multiplexing both reduce the
    probability of full retention, illustrating the tradeoff between mode count
    and channel robustness.}
    \label{fig:QMIMO2}
\end{figure}

\begin{figure}[t!]
    \centering
    \includegraphics[width = 0.475\textwidth]{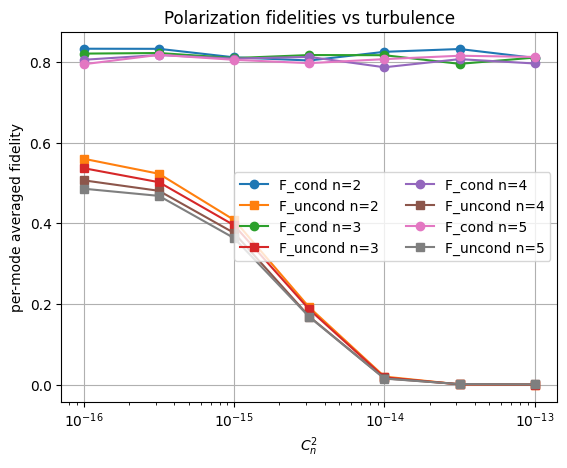}
    \caption{%
    \textbf{Average per-mode polarization fidelity versus turbulence.}
    Mean polarization fidelity of the logical channel as a function of turbulence
    strength $C_n^2$, shown for conditional (postselected, no-erasure) and
    unconditional (erasure-including) channels.
    Conditional fidelities remain high and largely independent of $n$,
    confirming polarization preservation by turbulence,
    while unconditional fidelities rapidly degrade as erasure dominates.}
    \label{fig:QMIMO3}
\end{figure}

\begin{figure}[t!]
    \centering
    \includegraphics[width = 0.475\textwidth]{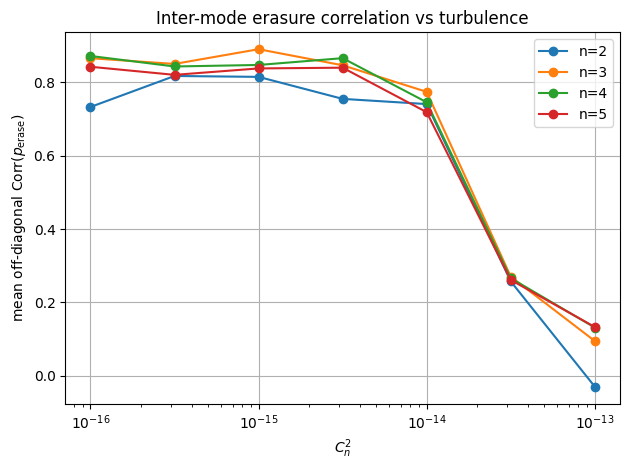}
    \caption{\textbf{Inter-Mode Loss Correlations.} Mean off-diagonal inter-mode erasure correlation as a function of the
turbulence strength $C_n^2$ for $n=2$–$5$ spatially multiplexed modes. The
correlation is computed from the joint statistics of per-rail erasure events and
quantifies the degree to which turbulence-induced loss acts as a common-mode
process across spatial channels. Strong correlations are observed in the weak
and moderate turbulence regimes, reflecting the shared propagation through a
single random medium, while correlations decrease at strong turbulence as
erasure probabilities approach unity and the loss process becomes saturated.}
    \label{fig:QMIMO4}
\end{figure}

The numerical results shown in Figs.~1–4 are obtained using a wave-optical,
split-step propagation model that directly resolves the evolution of spatial
modes through atmospheric turbulence. Polarization-encoded qubits are
multiplexed over $n=2$–$5$ co-propagating Laguerre–Gaussian modes with radial
index $p=0$ and azimuthal indices symmetrically chosen around $\ell=0$. The
optical wavelength is fixed at $\lambda=1550\,\mathrm{nm}$, the propagation
distance is $L=10\,\mathrm{km}$, and the beam waist at the transmitter is
$w_0=3\,\mathrm{cm}$. The transverse optical field is sampled on a
$128\times128$ Cartesian grid with spatial resolution
$\Delta x=2.5\,\mathrm{mm}$, which provides sufficient numerical support for
all considered spatial modes and ensures convergence of the Fresnel
propagation.

Atmospheric turbulence is modeled using a three-dimensional von K\'arm\'an
refractive-index spectrum with outer scale $L_0=30\,\mathrm{m}$ and inner scale
$l_0=5\,\mathrm{mm}$. The propagation path is divided into $K=40$ longitudinal
segments, each represented by a random phase screen. Longitudinal correlations
between successive screens are incorporated via an autoregressive AR(1)
process with correlation coefficient $\rho_z=0.9$, capturing the effects of
extended turbulence volumes rather than independent thin screens. The
turbulence strength is swept over
$C_n^2\in[10^{-16},10^{-13}]\,\mathrm{m}^{-2/3}$, spanning weak to strong
turbulence regimes relevant for kilometer-scale free-space links.  \color{black} Ensemble
statistics are obtained by averaging over $N_{\mathrm{MC}}=200$
independent turbulence realizations at each value of $C_n^2$. 
\color{black} 
While designing the propagation setup, we take the constraints in Ref. \cite{Rao:08}. All the constraints are satisfied in this study. 
\color{black}

For each realization, all input spatial modes are propagated independently
through the same sequence of phase screens using an angular-spectrum
(paraxial) split-step method. At the receiver plane, the propagated fields are
projected onto the finite set of kept spatial modes, yielding the intermodal
crosstalk matrix $T^\perp$. Coupling outside this basis is treated as leakage
to bath modes and constitutes erasure, while arrival in an incorrect kept port
is not regarded as an error. Distinguishable-photon statistics are computed
from the resulting single-photon coupling probabilities, whereas
indistinguishable-photon results include bosonic interference effects and are
evaluated using number-resolving detection statistics.

\textcolor{black}{Figure~\ref{fig:QMIMO1} shows the probability of collision (bunching) events within
the finite set of kept spatial ports. At weak turbulence, collisions are rare due
to near-ideal mode preservation. As turbulence increases to intermediate levels,
intermodal mixing within the kept subspace becomes significant while photon
survival remains appreciable, leading to a pronounced enhancement of collision
probability for indistinguishable photons due to bosonic many-body interference.
At strong turbulence, however, leakage into unobserved bath modes dominates and
the probability that multiple photons simultaneously remain within the kept
subspace rapidly vanishes. As a result, collision events are strongly suppressed,
producing the observed drop-off in both distinguishable and indistinguishable
cases.} Figure~\ref{fig:QMIMO2} shows the
probability that all photons remain within the kept set, demonstrating a
rapid, monotonic decay with increasing $C_n^2$ and multiplexing order $n$, and
the overlap of the distinguishable and indistinguishable
cases. Figure~\ref{fig:QMIMO3} reports per-mode polarization
fidelities, where conditional fidelities remain high across turbulence
strengths, consistent with scalar atmospheric propagation, while unconditional
fidelities decrease sharply as erasures become dominant. Finally,
Figure~\ref{fig:QMIMO4} presents the mean off-diagonal inter-mode erasure correlation as a
function of turbulence strength, revealing strong common-mode loss correlations
at weak-to-moderate turbulence and their gradual reduction in the strong
turbulence regime due to saturation of leakage. For each turbulence realization
$\omega$, spatial-mode projection induces a set of per-rail erasure probabilities
$\{\epsilon_m(\omega)\}_{m=1}^n$, defined via the leakage outside the kept spatial
subspace. Conditioned on a fixed realization, erasure events on different rails
are independent, but the shared turbulent medium renders the random variables
$\epsilon_m(\omega)$ statistically dependent across $m$. As a result, the
ensemble-averaged logical channel is characterized by a non-factorizable joint
erasure distribution $p(\mathbf{s})$, giving rise to correlated erasure events at
the logical level.

The correlation plotted in Fig.~4 quantifies this dependence by measuring the
mean off-diagonal correlation of erasure indicators across rails. In the weak-
and moderate-turbulence regimes, large-scale refractive-index fluctuations
dominate the propagation, producing common-mode wavefront distortions that
simultaneously affect all spatial modes. Consequently, fluctuations in the
per-rail erasure probabilities are strongly correlated, and erasure events tend
to occur jointly across multiple rails. In contrast, as the turbulence strength
increases, leakage into bath modes becomes nearly deterministic, with
$\epsilon_m(\omega)\to 1$ for all $m$. In this strong-turbulence regime, the
variance of the erasure indicators collapses, suppressing measurable
correlations despite the continued presence of a shared medium. The decay of
inter-mode erasure correlation observed in Fig.~4 therefore reflects saturation
of loss rather than a reduction in physical coupling, and is a generic feature
of correlated erasure channels derived from strongly lossy propagation.

\section{Conclusions}
We have developed a first-principles framework for modeling Quantum MIMO channels arising in spatially multiplexed free-space optical links. Starting from wave-optical propagation through three-dimensional atmospheric turbulence, the model explicitly connects physical effects—intermodal coupling, finite-aperture detection, and spatial-mode projection—to the resulting logical quantum channel. This approach provides a direct bridge between  optical propagation and finite-dimensional quantum information models.
By distinguishing between distinguishable and indistinguishable photon regimes, we showed that spatial multiplexing can induce fundamentally different noise structures, ranging from single-photon Kraus maps to intrinsically many-body interference effects governed by matrix permanents. To obtain a completely positive and trace-preserving description on a fixed logical space, we introduced an erasure-extended encoding that captures turbulence-induced leakage and photon loss while preserving the physical interpretation of spatial mode mixing.
Within this framework, we demonstrated that, under physically standard assumptions, the dominant logical impairment in atmospheric Quantum MIMO links is naturally described as a correlated $n$-qubit erasure channel, with correlations arising from the shared turbulent medium. Correlated Pauli noise models emerge only as limiting or approximate descriptions, clarifying their domain of validity for free-space quantum communication. The resulting formulation enables systematic investigation of correlated errors, partial erasures, and their impact on multi-qubit protocols.
The framework developed here establishes a foundation for analyzing quantum error mitigation and coding strategies tailored to spatially multiplexed free-space links \cite{koudia2025crosstalk,koudia2023quantum}. Extensions to polarization-dependent effects, adaptive mode control, and experimental validation of correlated erasure statistics constitute promising directions for future work.

\bibliographystyle{IEEEtran}
\bibliography{ref}

@article{koudia2019superposition,
  title={Superposition of causal orders for quantum discrimination of quantum processes},
  author={Koudia, Seid and Gharbi, Abdelhakim},
  journal={International Journal of Quantum Information},
  volume={17},
  number={07},
  pages={1950055},
  year={2019},
  publisher={World Scientific}
}

@ARTICLE{11317988,
  author={Koudia, Seid and Oleynik, Leonardo and Bayraktar, Mert and Rehman, Junaid Ur and Chatzinotas, Symeon},
  journal={IEEE Communications Surveys \& Tutorials}, 
  title={Physical Layer Aspects of Quantum Communications: A Survey}, 
  year={2025},
  volume={},
  number={},
  pages={1-1},
  doi={10.1109/COMST.2025.3647980}}

@article{koudia2025spatial,
  title={Spatial-mode diversity and multiplexing for continuous variables quantum communications},
  author={Koudia, Seid and Oleynik, Leonardo and Ur Rehman, Junaid and Chatzinotas, Symeon},
  journal={Communications Physics},
  volume={8},
  number={1},
  pages={351},
  year={2025},
  publisher={Nature Publishing Group UK London}
}

@article{ur2025mimo,
  title={MIMO Quantum Communication in Correlated Pure-Loss Channels},
  author={ur Rehman, Junaid and Rizvi, Syed Muhammad Abuzar and Koudia, Seid and Chatzinotas, Symeon and Shin, Hyundong},
  journal={IEEE Communications Letters},
  year={2025},
  publisher={IEEE}
}

@article{pirandola2021limits,
  title={Limits and security of free-space quantum communications},
  author={Pirandola, Stefano},
  journal={Physical Review Research},
  volume={3},
  number={1},
  pages={013279},
  year={2021},
  publisher={APS}
}

@book{AndrewsPhillips200,
  author    = {L. C. Andrews and R. L. Phillips},
  title     = {Laser Beam Propagation through Random Media},
  publisher = {SPIE Press},
  edition   = {2},
  year      = {2005},
  chapter   = {3},
  doi       = {10.1117/3.547620}
}

@book{GoodmanFourierOptics,
  author    = {Joseph W. Goodman},
  title     = {Introduction to Fourier Optics},
  publisher = {W. H. Freeman},
  edition   = {4},
  year      = {2017},
  chapter   = {4}
}

@article{Rao:08,
author = {Ruizhong Rao},
journal = {Appl. Opt.},
number = {2},
pages = {269--276},
publisher = {Optica Publishing Group},
title = {Statistics of the fractal structure and phase singularity of a plane light wave            propagation in atmospheric turbulence},
volume = {47},
month = {Jan},
year = {2008},
url = {https://opg.optica.org/ao/abstract.cfm?URI=ao-47-2-269},
doi = {10.1364/AO.47.000269},
}

@article{fabre2020modes,
  title={Modes and states in quantum optics},
  author={Fabre, Claude and Treps, Nicolas},
  journal={Reviews of Modern Physics},
  volume={92},
  number={3},
  pages={035005},
  year={2020},
  publisher={APS}
}

@article{koudia2023quantum,
  title={The quantum internet: an efficient stabilizer states distribution scheme},
  author={Koudia, Seid},
  journal={Physica Scripta},
  volume={99},
  number={1},
  pages={015115},
  year={2023},
  publisher={IOP Publishing}
}

@article{junaid2025diversity,
  title={Diversity and multiplexing in quantum MIMO channels},
  author={ur Rehman, Junaid and Oleynik, Leonardo and Koudia, Seid and Bayraktar, Mert and Chatzinotas, Symeon},
  journal={EPJ Quantum Technology},
  volume={12},
  number={1},
  pages={18},
  year={2025},
  publisher={Springer Berlin Heidelberg}
}

@article{koudia2025crosstalk,
  title={Crosstalk-resilient quantum MIMO for scalable quantum communications},
  author={Koudia, Seid and Chatzinotas, Symeon},
  journal={npj Quantum Information},
  volume={11},
  number={1},
  pages={162},
  year={2025},
  publisher={Nature Publishing Group UK London}
}

@INPROCEEDINGS {pengqce,
author = { Peng, Heyang and Koudia, Seid and Oleynik, Leonardo and Chatzinotas, Symeon },
booktitle = { 2025 IEEE International Conference on Quantum Computing and Engineering (QCE) },
title = {{ Performance Analysis of MDI-QKD in Thermal-Loss and Phase Noise Channels }},
year = {2025},
volume = {},
ISSN = {},
pages = {976-983},
doi = {10.1109/QCE65121.2025.00109},
url = {https://doi.ieeecomputersociety.org/10.1109/QCE65121.2025.00109},
publisher = {IEEE Computer Society},
address = {Los Alamitos, CA, USA},
month =sep}

@article{peng2025performance,
  title={Performance Analysis of MDI-QKD in Thermal-Loss and Phase Noise Channels},
  author={Peng, Heyang and Koudia, Seid and Oleynik, Leonardo and Chatzinotas, Symeon},
  journal={arXiv preprint arXiv:2504.16561},
  year={2025}
}

@book{Ishimaru1999,
  author    = {A. Ishimaru},
  title     = {Wave Propagation and Scattering in Random Media},
  publisher = {IEEE Press},
  year      = {1999}
}

@article{Allen1992,
  author  = {L. Allen and M. W. Beijersbergen and R. J. C. Spreeuw and J. P. Woerdman},
  title   = {Orbital angular momentum of light and the transformation of Laguerre--Gaussian laser modes},
  journal = {Physical Review A},
  volume  = {45},
  pages   = {8185--8189},
  year    = {1992},
  doi     = {10.1103/PhysRevA.45.8185}
}

@article{Willner2015,
  author  = {A. E. Willner and others},
  title   = {Optical communications using orbital angular momentum beams},
  journal = {Advances in Optics and Photonics},
  volume  = {7},
  number  = {1},
  pages   = {66--106},
  year    = {2015},
  doi     = {10.1364/AOP.7.000066}
}

@book{Wilde2017,
  author    = {M. M. Wilde},
  title     = {Quantum Information Theory},
  publisher = {Cambridge University Press},
  edition   = {2},
  year      = {2017},
  doi       = {10.1017/9781316809976}
}

@book{Goldsmith2005,
  author    = {A. Goldsmith},
  title     = {Wireless Communications},
  publisher = {Cambridge University Press},
  year      = {2005},
  doi       = {10.1017/CBO9780511841224}
}
\end{document}